\documentclass[runningheads]{llncs}
\usepackage{graphicx}
\usepackage{amsfonts}
\usepackage{amsmath}

\newcommand{\argmin}{\mathop{\mathrm{argmin}}}
\newcommand{\norm}[1]{\left\Vert#1\right\Vert}

\begin{document}
\title{Consensus Neural Network for Medical Imaging Denoising with Only Noisy Training Samples}
\titlerunning{Consensus Network}
%
\author{Dufan Wu\inst{1,2} \and
Kuang Gong\inst{1,2} \and
Kyungsang Kim\inst{1,2} \and
Quanzheng Li\inst{1,2}
}
\authorrunning{Dufan Wu et al.}
%
\institute{\inst{} Center for Advanced Medical Computing and Analysis, Massachusetts General Hospital and Harvard Medical School, Boston, MA 02114, USA \and
\inst{} Gordon Center for Medical Imaging, Massachusetts General Hospital and Harvard Medical School, Boston, MA 02114, USA}
\maketitle              
\begin{abstract}
Deep neural networks have been proved efficient for medical image denoising. Current training methods require both noisy and clean images. However, clean images cannot be acquired for many practical medical applications due to naturally noisy signal, such as dynamic imaging, spectral computed tomography, arterial spin labeling magnetic resonance imaging, etc. In this paper we proposed a training method which learned denoising neural networks from noisy training samples only. Training data in the acquisition domain was split to two subsets and the network was trained to map one noisy set to the other. A consensus loss function was further proposed to efficiently combine the outputs from both subsets. A mathematical proof was provided that the proposed training scheme was equivalent to training with noisy and clean samples when the noise in the two subsets was uncorrelated and zero-mean. The method was validated on Low-dose CT Challenge dataset and NYU MRI dataset and achieved improved performance compared to existing unsupervised methods. 

\end{abstract}
\section{Introduction}
Deep neural network has been proved efficient for noise and artifacts reduction in medical imaging reconstruction \cite{jin2017deep}\cite{kang2017deep}. Despite of the superior image quality achieved by neural networks compared to handcrafted prior functions \cite{sidky2008image}, almost all the neural network based methods require both noisy and clean images during the training. However, such clean images are not always accessible due to naturally noisy signals acquired in many medical applications. Dynamic imaging, including dynamic positron emission tomography (PET), dynamic magnetic resonance (MR), and computed tomography (CT) perfusion, acquires signals with rapid temporal change, and the signal quality is limited due to short acquisition time. Spectral CT has very noisy material images due to the ill-posed decomposition procedure \cite{Niu2014Iterative}. Arterial spin labeling (ASL) MR also has noisy images due to the low efficiency in labeling arterial blood with magnetic field, which results in noisy signals emitted by the labeled blood \cite{Bibic2010denoising}. 

A recent work proposed in 2018, Noise2noise \cite{pmlr-v80-lehtinen18a}, demonstrated that denoising networks can be learned by mapping a noisy image to another noisy realization of the same image, and the performance was similar to that using noisy-clean pairs. However, it requires at least two noise realizations for each training sample, which is not readily available for most medical images. Even if two noise realizations are given, Noise2noise framework can only effectively use one of them, which degraded the achievable image quality. Last but not least, \cite{pmlr-v80-lehtinen18a} did not clarify the conditions for Noise2noise to work, which can be problematic for medical imaging due to the complicated noise characteristics. 

In this work we proposed a consensus network for medical imaging which required only noisy data for training. Our consensus network was inspired by Noise2noise but with major improvements for medical imaging. The Noise2noise framework was first analyzed with a newly proposed mathematical theorem, which further clarified its applicable condition. Based on conditions derived from the theorem, the acquired signals were split to two sets to reconstruct images with different noise realizations. The denoising network was then trained to map one noise realization to the other, with a novel loss function which efficiently aggregated both noise realizations during testing time. The proposed method was evaluated on Low-dose CT (LDCT) Challenge dataset \cite{aapm2017low} for quarter-dose CT image denoising, and New York University (NYU) MR dataset \cite{hammernik2018learning} for 4$\times$ undersampling parallel imaging. Results from both datasets demonstrated improved performance compared to the original Noise2noise framework and iterative reconstruction methods. 

\section{Methodology}
We will first provide proof for Noise2Noise training, then derive loss function for the proposed consensus network based on our theorem. 
\subsection{Noise2noise Training}
Given paired clean and noisy images $\mathbf{x}_i, \mathbf{x}_i + \mathbf{n}_i \in \mathbb{R}^n$, conventional method to train denoising network under L2-loss is:
\begin{equation}\label{noise2clean}
\mathbf{\Theta}_c = \argmin_\mathbf{\Theta}\frac{1}{N}\sum_{i=1}^N\norm{f(\mathbf{x}_i + \mathbf{n}_i;\mathbf{\Theta}) - \mathbf{x}_i}_2^2,
\end{equation}
where $f(\mathbf{x};\mathbf{\Theta}): \mathbb{R}^n \rightarrow \mathbb{R}^n$ is the denoising neural network and $\mathbf{\Theta} \in \mathbb{R}^m$ is the parameters to be trained. Equation (\ref{noise2clean}) is referred as Noise2clean training. 

Noise2noise training uses two independent noise realizations of each sample for the training:
\begin{equation}\label{noise2noise}
\mathbf{\Theta}_n = \argmin_\mathbf{\Theta}\frac{1}{N}\sum_{i=1}^N\norm{f(\mathbf{x}_i + \mathbf{n}_{i1};\mathbf{\Theta}) - (\mathbf{x}_i+\mathbf{n}_{i2})}_2^2,
\end{equation}
where $\mathbf{n}_{i1}, \mathbf{n}_{i2} \in \mathbb{R}^n$ are two independent noise samples. 

The equivalence between Noise2noise (\ref{noise2noise}) and Noise2clean (\ref{noise2clean}) is guaranteed by theorem \ref{theorem_eq}, which is one of our main contributions:
\begin{theorem}\label{theorem_eq}
If conditional expectation $E\{\mathbf{n}_{i2}| \mathbf{x}_i+\mathbf{n}_{i1}\} = \mathbf{0}$ $\forall$ $\mathbf{\Theta} \in \mathbb{R}^m$ and $i$, then $\lim_{N\rightarrow\infty}\mathbf{\Theta}_n = \mathbf{\Theta}_c$.
\end{theorem}

\begin{proof}
Let $\mathbf{y}_i := f(\mathbf{x}_i + \mathbf{n}_{i1};\mathbf{\Theta})$ and expand the loss function in (\ref{noise2noise}):
\begin{equation}\label{expand}
\begin{split}
&\frac{1}{N}\sum_{i=1}^N\norm{\mathbf{y}_i - (\mathbf{x}_i+\mathbf{n}_{i2})}_2^2 \\
= &\frac{1}{N}\sum_{i=1}^N\norm{\mathbf{y}_i - \mathbf{x}_i}_2^2 - \frac{1}{N}\sum_{i=1}^N 2\mathbf{n}_{i2}^T\mathbf{y}_i + \frac{1}{N}\sum_{i=1}^N(\mathbf{n}_{i2}^T\mathbf{n}_{i2} + 2\mathbf{n}_{i2}^T\mathbf{x}_i)
\end{split}
\end{equation}
The first term is Noise2clean loss (\ref{noise2clean}) and the last term is irrelevant to $\mathbf{\Theta}$. 

For the second term, according to Lindeberg-Levy central limit theorem:
\begin{equation}\label{clt}
\sqrt{N}( \frac{1}{N}\sum 2\mathbf{n}_{i2}^T\mathbf{y}_i - E\{ 2\mathbf{n}_{i2}^T\mathbf{y}_i \}) \xrightarrow{d} N(0, Var\{ 2\mathbf{n}_{i2}^T\mathbf{y}_i \}),
\end{equation}
where $Var\{ x \}$ is the variance of $x$, and $N(0, \sigma^2)$ is a normal distribution with variance $\sigma^2$. As $Var\{ 2\mathbf{n}_{i2}^T\mathbf{y}_i \}$ is finite, the second term in (\ref{expand}) will converge to $E\{ 2\mathbf{n}_{i2}^T\mathbf{y}_i \}$ as $N\rightarrow\infty$.

The expectation can be written as conditional expectation:
\begin{equation}\label{condE}
E\{2\mathbf{n}_{i2}^T\mathbf{y}_i\} = 2E\{E^T\{\mathbf{n}_{i2} | \mathbf{y}_i\} \mathbf{y}_i\} = 2E\{E^T\{\mathbf{n}_{i2} | \mathbf{x}_i + \mathbf{n}_{i1}\} f(\mathbf{x}_i+\mathbf{n}_{i1};\mathbf{\Theta})\},
\end{equation}
the last equivalence was due to that $\mathbf{y}_i = f(\mathbf{x}_i+\mathbf{n}_{i1};\mathbf{\Theta})$ was deterministic.

Equation (\ref{condE}) equals to 0 given the theorem's assumption, which lead to diminishing second term in (\ref{expand}) as $N \rightarrow \infty$. As the third term was irrelevant to $\mathbf{\Theta}$, we have:
\begin{equation}
\argmin_\mathbf{\Theta}\frac{1}{N}\sum_{i=1}^N\norm{\mathbf{y}_i - \mathbf{x}_i}_2^2 = \argmin_\mathbf{\Theta}\frac{1}{N}\sum_{i=1}^N\norm{\mathbf{y}_i - (\mathbf{x}_i+\mathbf{n}_{i2})}_2^2,
\end{equation}
as $N \rightarrow \infty$, which implies $\mathbf{\Theta}_n = \mathbf{\Theta}_c$. $\square$
\end{proof}

\subsection{Consensus Loss Function}
The key to Noise2noise training is finding independent and zero-mean noise realizations $\mathbf{n}_{i1}$ and $\mathbf{n}_{i2}$. For medical imaging, it can be achieved by splitting the measurement to independent sets which are reconstructed separately. For example, in low-dose CT,  $\mathbf{x}_i+\mathbf{n}_{i1}$ and $\mathbf{x}_i+\mathbf{n}_{i2}$ can be images reconstructed from odd and even projections respectively. 

The major drawback of this ``data splitting" method was that both $\mathbf{x}_i+\mathbf{n}_{i1}$ and $\mathbf{x}_i+\mathbf{n}_{i2}$ were noisier than the original noisy image, which restricted the quality of denoised images. To efficiently aggregate both noise realizations, they should be taken into consideration in the same loss function. According to theorem \ref{theorem_eq}, such loss function can be designed under the Noise2clean framework first, then substituting clean images $\mathbf{x}_i$ with its noisy version when appropriate. The following Noise2clean model was considered to derive our consensus loss:
\begin{equation}
\begin{split}
L_c & = \frac{1}{N}\sum_{i=1}^N\norm{\frac{f(\mathbf{x}_{i}+\mathbf{n}_{i1};\mathbf{\Theta}_1) + f(\mathbf{x}_{i}+\mathbf{n}_{i2};\mathbf{\Theta}_2)}{2} - \mathbf{x}_i}_2^2 \\
       & = \frac{1}{N}\sum_{i=1}^N\left\{ \frac{1}{2}\norm{\mathbf{y}_{i1} - \mathbf{x}_i}_2^2 + \frac{1}{2}\norm{\mathbf{y}_{i2} - \mathbf{x}_i}_2^2 - \frac{1}{4}\norm{\mathbf{y}_{i1} - \mathbf{y}_{i2}}_2^2  \right\}, 
\end{split}
\end{equation}
where $\mathbf{y}_{i1} := f(\mathbf{x}_{i}+\mathbf{n}_{i1};\mathbf{\Theta}_1)$ and $\mathbf{y}_{i2} := f(\mathbf{x}_{i}+\mathbf{n}_{i2};\mathbf{\Theta}_2)$. $\mathbf{\Theta}_1, \mathbf{\Theta}_2 \in \mathbb{R}^m$ were the trainable variables. The second equivalence was just simple factorization. This model effectively aggregated $\mathbf{x}_i+\mathbf{n}_{i1}$ and $\mathbf{x}_i+\mathbf{n}_{i2}$ by taking both of them into consideration in the same loss function.

Similar to theorem \ref{theorem_eq}, $\mathbf{x}_{i}$ can be substituted with $\mathbf{x}_i + \mathbf{n}_{i1}$ or $\mathbf{x}_i + \mathbf{n}_{i2}$ during training, and our Noise2noise consensus loss is given as:
\begin{equation}\label{consensus_loss}
\begin{split}
L_n = \frac{1}{N}\sum_{i=1}^N\biggl\{ &\frac{1}{2}\norm{f(\mathbf{x}_{i}+\mathbf{n}_{i1};\mathbf{\Theta}_1) - (\mathbf{x}_i + \mathbf{n}_{i2})}_2^2  + \\
						&\frac{1}{2}\norm{f(\mathbf{x}_{i}+\mathbf{n}_{i2};\mathbf{\Theta}_2) - (\mathbf{x}_i + \mathbf{n}_{i1})}_2^2 - \\
						&\frac{1}{4}\norm{f(\mathbf{x}_{i}+\mathbf{n}_{i1};\mathbf{\Theta}_1) - f(\mathbf{x}_{i}+\mathbf{n}_{i2};\mathbf{\Theta}_2)}_2^2\biggr\}
\end{split}
\end{equation}
and the denoised image is given by:
\begin{equation}\label{recon_image}
\mathbf{z}_i = \frac{ f(\mathbf{x}_{i}+\mathbf{n}_{i1};\mathbf{\Theta}_1) + f(\mathbf{x}_{i}+\mathbf{n}_{i2};\mathbf{\Theta}_2)}{2}
\end{equation}

\subsection{Regularization}\label{sec_regularization}
In practice, $\mathbf{n}_{i1}$ and $\mathbf{n}_{i2}$ could be correlated due to aliasing related to structures, and $N$ may not be large enough in mini-batch training. Two regularization terms were added beside $L_n$ for artifacts reduction and detail preservation.

The first term was weight decay defined as:
\begin{equation}
L_w = \norm{\mathbf{\Theta}_1}_2^2 + \norm{\mathbf{\Theta}_2}_2^2,
\end{equation}
which was effective at eliminating artifacts due to noise correlation. 

The second term was image consistency:
\begin{equation}\label{eq_consist}
L_r = \frac{1}{N}\sum_{i=1}^N\norm{\mathbf{z}_i - \mathbf{x}_i^{est}}_2^2,
\end{equation}
where $\mathbf{z}_i$ is given by (\ref{recon_image}) and $\mathbf{x}_i^{est}$ is an estimation of $\mathbf{x}_i$. It could be the original noisy images, or reconstructed with artifacts-suppressing algorithms. 

The final loss function for training was
\begin{equation}
L = L_n + \beta_w L_w + \beta_r L_r,
\end{equation}
where $\beta_w$ could be tuned first with $\beta_r=0$ to remove artifacts, then $\beta_r$ could be tuned for image details with fixed $\beta_w$. The denoised image was given by (\ref{recon_image}).

\section{Experiments}

\subsection{Data Preparation}
We validated the proposed methods on two datasets: LDCT Grand Challenge \cite{aapm2017low} and NYU MR images \cite{hammernik2018learning}. Both datasets have high quality images which provided reference for evaluation. 

The LDCT dataset consisted of abdomen CT scans from 10 patients. Quarter-dose raw data was synthesized from the acquisitions by realistic noise insertion. The raw data was rebined to multi-slice fanbeam sinogram for image reconstruction. We randomly chose 50 slices from each patients for our study, and used 8 patients for training with the other 2 patients for testing. For each slice, we split the 2304 quarter-dose projections to odd and even set and used filtered backprojection (FBP) with Hann filter to reconstruct the $\mathbf{x}_i+\mathbf{n}_{i1}$ and $\mathbf{x}_i+\mathbf{n}_{i2}$ required for the training of consensus network.

The NYU MR dataset used in this study consisted of sagittal knee MR scans with Turbo Spin Echo sequence from 20 patients. 4$\times$ catesian downsampling was synthesized by downsampling the original kspace data with different random masks for each slice. The central 48 lines were kept for aliasing reduction and sensitivity map estimation. Each patients had 31 to 35 slices and we randomly chose 16 patients for training with 4 patients for testing. For each patient, the sampled lines were further randomly split to two sets with 8$\times$ dwonsampling each. The central 48 lines were kept for both sets. 

Zero-filling was used to reconstruct $\mathbf{x}_i+\mathbf{n}_{i1}$ and $\mathbf{x}_i+\mathbf{n}_{i2}$ due to its linear property so that the artifacts was zero-mean. The random sampling part of the kspace data was also amplified by a factor of 8 to enforce the zero-mean property. 

\subsection{Parameters}
We used UNet \cite{jin2017deep}\cite{Ronneberger2015unet} as $f(\mathbf{x}; \mathbf{\Theta})$ in all the studies. 32 and 64 basic featuremaps were used for LDCT and NYU MR datasets respectively. For MR study, the real and image part of the images were fed to the network as two channels. The $\mathbf{x}_i^{est}$ in CT and MR studies were quarter-dose FBP results and SENSE  \cite{pruessmann1999sense} results respectively. The hyperparameters, $\beta_w$ and $\beta_r$, were tuned according to section \ref{sec_regularization}. LDCT study used $\beta_w = 5\times10^{-6}$ and $\beta_r=0.5$. MR study used $\beta_w=1\times10^{-6}$ and $\beta_r=5$.

The networks were trained on $96\times 96$ local patches with mini-batch size of 40. At each epoch, 40 patches were randomly extracted from each training slice. The networks were trained by Adam algorithm with learning rate of $10^{-4}$ for 100 epochs. CT images were normalized to HU / 1000 and MR images were normalized to [-1, 1] before being fed to the networks. 

\section{Results}

The root mean squares error (RMSE) and structural similarity index (SSIM) on the testing quarter-dose CT dataset are given in table \ref{tab_ct}, and one of the testing slices with two lesions is given in figure \ref{fig_ct}. Beside the proposed method, we also gave the results from quarter-dose FBP, iterative reconstruction with total variation (TV) minimization \cite{sidky2008image}, original Noise2noise (\ref{noise2noise}) and Noise2clean (\ref{noise2clean}). All the methods did not require clean images for training or testing except for the Noise2clean, which was supervised learning. 

\begin{figure}[t]
\begin{center}
\includegraphics[width=0.8\textwidth]{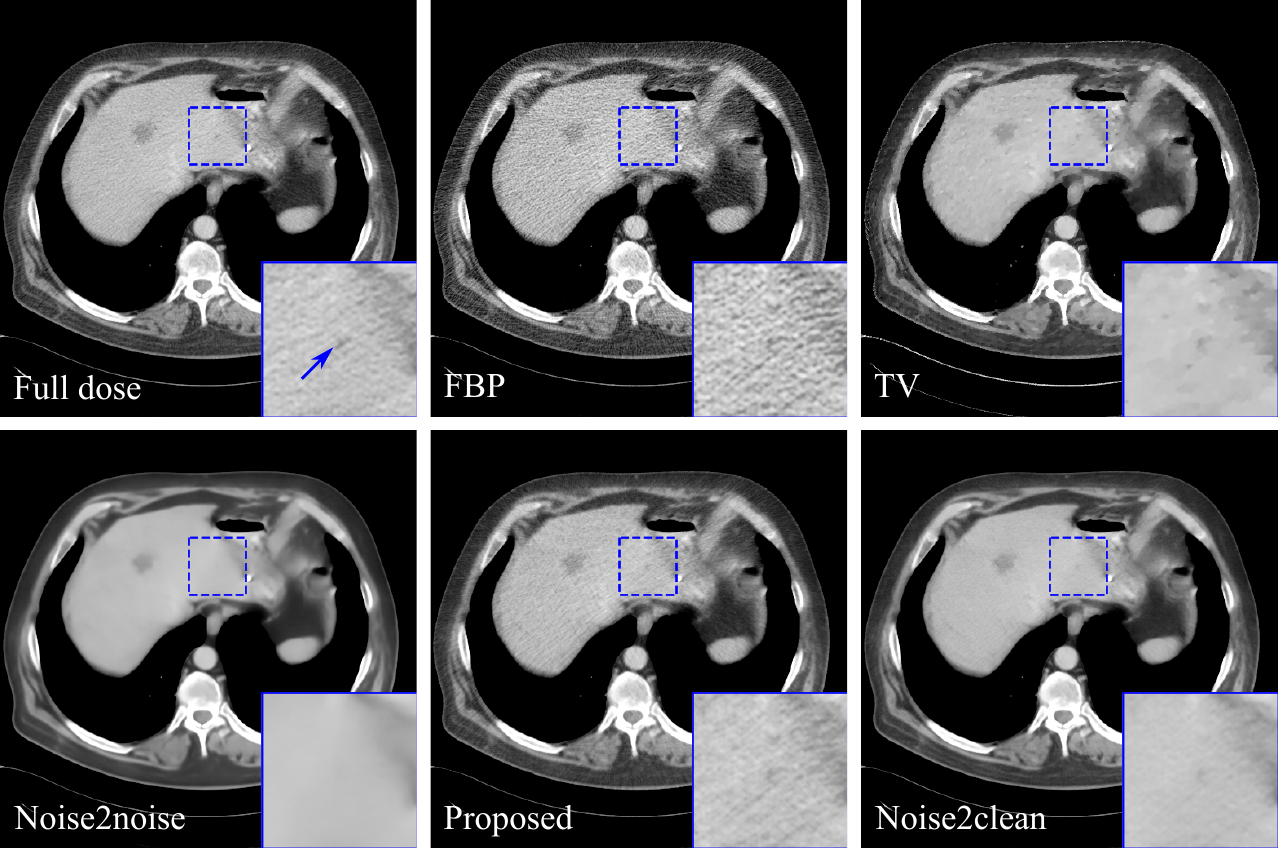}
\end{center}
\caption{A testing CT slice processed with different methods. A small lesion is marked with blue arrow on the full-dose image. The display window is [-160, 240] HU.} \label{fig_ct}
\end{figure}

\begin{table}[t]
\caption{RMSEs and SSIMs of the testing CT images. SSIMs were calculated within the liver window [-160, 240] HU.
Noise2clean required clean images for training.}\label{tab_ct}
\begin{center}
\begin{tabular}{p{2cm}p{1.8cm}p{1.8cm}p{1.8cm}p{1.8cm}p{1.8cm}}
\hline
Index & FBP & TV & Noise2noise & Proposed & Noise2clean\\
\hline
RMSE(HU) & $27.6\pm5.3$ & $36.9\pm4.4$ & $17.4\pm3.2$ & $18.4\pm2.9$ & $14.7\pm2.8$ \\
SSIM & $81.5\pm4.7$ & $81.8\pm3.5$ & $84.2\pm4.6$ & $87.1\pm3.4$ & $87.7\pm3.3$ \\
\hline
\end{tabular}
\end{center}
\end{table}

The proposed method achieved the best SSIM among all the unsupervised methods. Although Noise2noise had lower RMSE compared to the proposed method, its images were significantly oversmoothed as shown in figure \ref{fig_ct}. The higher RMSE of the proposed methods were mainly due to mismatch in noise patterns in the reference images and the proposed results. The proposed method preserved both the small lesion structure and textures compared to TV and Noise2noise. The structural details recovered by the proposed method were very similar to that by the Noise2clean training, and they achieved close SSIMs. 

The 4$\times$ downsampling MR testing results are given in table \ref{tab_mr} and figure \ref{fig_mr}. Results from Zero-filling, SENSE, Noise2noise and Noise2clean are given beside the proposed method. 

\begin{figure}[t]
\begin{center}
\includegraphics[width=0.8\textwidth]{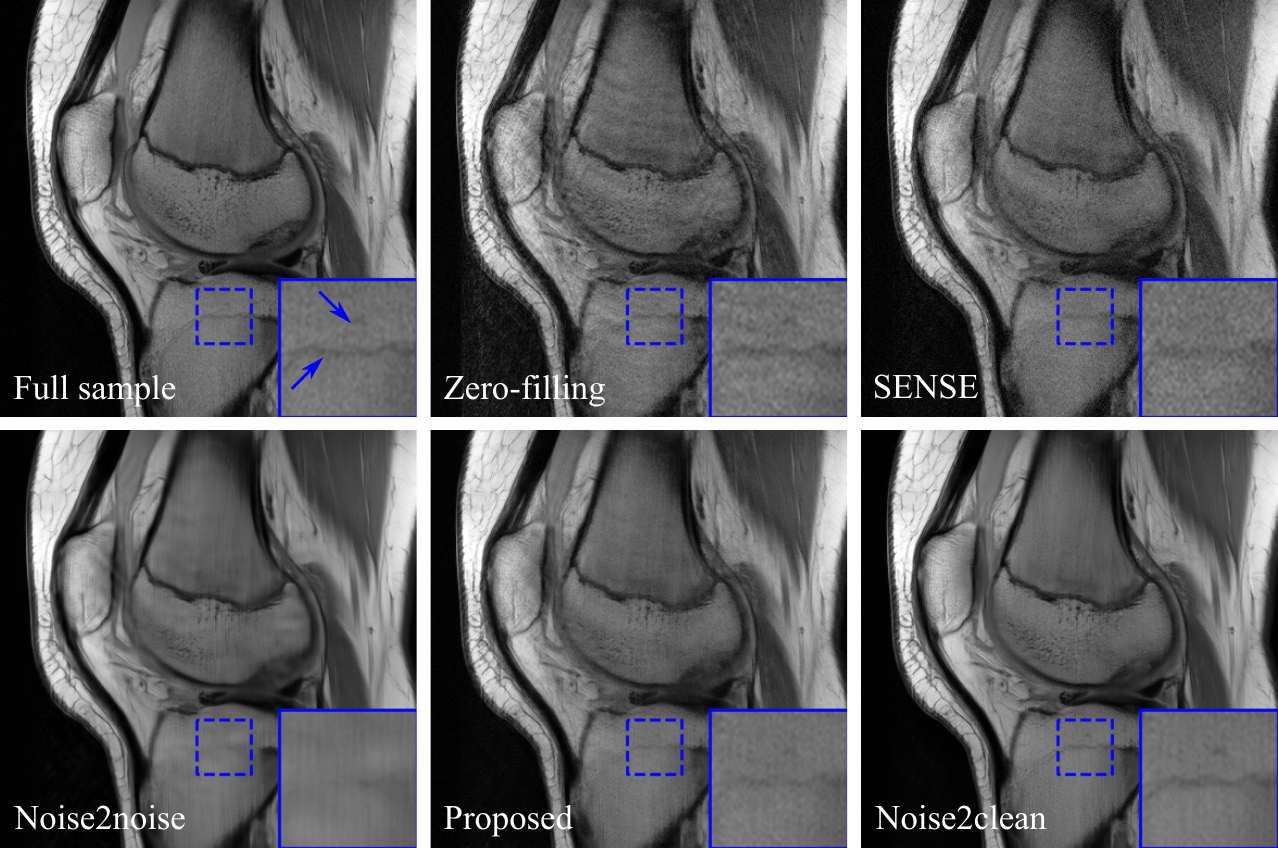}
\end{center}
\caption{Amplitude images of a testing MR slice processed with different methods. Structures are marked with blue arrows on the fully sampled image.} \label{fig_mr}
\end{figure}

\begin{table}[t]
\caption{RMSEs and SSIMs of the testing MR images. SSIMs were calculated for the amplitude images.
Noise2clean required clean images for training.}\label{tab_mr}
\begin{center}
\begin{tabular}{p{2cm}p{1.8cm}p{1.8cm}p{1.8cm}p{1.8cm}p{1.8cm}}
\hline
Index & Zero-filling & SENSE & Noise2noise & Proposed & Noise2clean\\
\hline
RMSE($10^{-3})$ & $36.3\pm7.6$ & $38.4\pm4.6$ & $27.7\pm6.3$ & $23.4\pm4.2$ & $21.0\pm4.5$ \\
SSIM & $90.1\pm1.9$ & $86.5\pm1.1$ & $94.6\pm1.3$ & $95.2\pm0.6$ & $96.4\pm0.8$ \\
\hline
\end{tabular}
\end{center}
\end{table}

In the MR study, the proposed method achieved the best RMSE and SSIM among all the unsupervised methods. Noise2noise results were still oversmoothed. The proposed method preserved detailed structures of the knee compared to Noise2noise, and had significantly reduced artifacts and noise compared to zero-filling and SENSE. Compared to Noise2clean result, the proposed method recovered almost the same amount of structures despite of some slight loss in contrast. The artifacts and noise level were also similar for the two results. 

\section{Conclusion and Discussion}
In this paper we proposed an unsupervised learning method for medical image denoising which only required noisy samples during training. A novel theorem was proposed for the Noise2noise framework. Our consensus network was proposed based on the theorem with novel framework and loss functions designed for medical imaging. The proposed loss function efficiently utilized both noise realizations of the same object and achieved improved performance compared to Noise2noise under the same framework. The proposed method achieved better performance than other unsupervised methods. Its image quality was close to that of supervised denoising networks. 

The proposed method had weak assumption on the property of noise and images. It worked for both local noise in CT and non-local undersampling artifacts in MR. The noise was only required to be zero-mean and independent in the two realizations. Whereas zero-mean can be guaranteed with appropriate reconstruction algorithms, the independence property could be breached due to various factors. The artifacts caused by correlations were successfully compensated with the proposed regularizations. 

Although the proposed method need to split the raw data to create the training dataset, its deployment can work in image domain only. After the consensus network was trained, another network could be trained to map noisy images to the output of consensus network in a conventional supervised manor. Hence, the networks' deployment do not need to interfere existing workflow of scanners. 

We achieved promising results on existing public datasets where high-quality reference images were available, which provided reliable evaluation of the method's performance. The method can be applied to more challenging applications where only noisy images are available, such as dynamic PET, ASL MR, and spectral CT. Furthermore, the method itself could also be improved, such as replacing the image consistency (\ref{eq_consist}) with data consistency; or replacing UNet with reconstruction networks.

%
%
%
%

\end{document}